\documentclass[a4paper]{article}
\usepackage[utf8]{inputenc}
\usepackage{amsmath}
\usepackage{amsthm}
\usepackage{amssymb}
\usepackage{amsfonts}
\usepackage{authblk}
\usepackage{mathtools}
\usepackage{float}

\newtheorem{theorem}{Theorem}
\newtheorem{proposition}{Proposition}
\newtheorem{lemma}{Lemma}

\newtheorem{definition}{Definition}

\newcommand{\eps}{\ensuremath{\varepsilon}}
\newcommand{\abs}[1]{\ensuremath{\mathopen\lvert #1 \mathclose\rvert}}
\newcommand{\norm}[1]{\ensuremath{\mathopen\lVert #1 \mathclose\rVert}}
\newcommand{\Abs}[1]{\ensuremath{\left| #1 \right|}}

\newcommand{\RR}{\ensuremath{\mathbb{R}}}
\newcommand{\NN}{\ensuremath{\mathbb{N}}}

\newcommand{\deltasel}{\ensuremath{\operatorname{Sel}_\delta}}
\newcommand{\perm}[1]{\ensuremath{\mathit{Perm}(#1)}}

\renewcommand{\P}{\ensuremath{\mathrm{P}}}

\newcommand{\PPAD}{\ensuremath{\mathrm{PPAD}}}
\newcommand{\FIXP}{\ensuremath{\mathrm{FIXP}}}

\newcommand{\subslash}[1]{\ensuremath{/_{\mkern-5mu#1}}}

\begin{document}

\title{Computational Complexity of Computing a Quasi-Proper Equilibrium}
\author[1]{Kristoffer Arnsfelt Hansen\thanks{Supported by the Independent Research Fund Denmark under grant no. 9040-00433B.}}
\author[2]{Troels Bjerre Lund}
\affil[1]{Aarhus University}
\affil[2]{IT-University of Copenhagen}
\date{}
\maketitle
\begin{abstract}
  We study the computational complexity of computing or approximating
  a quasi-proper equilibrium for a given finite extensive form game of
  perfect recall. We show that the task of computing a symbolic
  quasi-proper equilibrium is $\PPAD$-complete for two-player
  games. For the case of zero-sum games we obtain a polynomial time
  algorithm based on Linear Programming. For general $n$-player games
  we show that computing an approximation of a quasi-proper
  equilibrium is $\FIXP_a$-complete. Towards our results for
  two-player games we devise a new perturbation of the strategy space
  of an extensive form game which in particular gives a new proof of
  existence of quasi-proper equilibria for general $n$-player games.
\end{abstract}

\section{Introduction}

A large amount of research has gone into
defining~\cite{Selten65,IJGT:Selten75,IJGT:Myerson78,E:KrepsW82,IJGT:vanDamme84}
and
computing~\cite{E:vonStengel02,ET:MiltersenSorensen10,SAGT:HansenMS10,AAAI:FarinaGG17,EC:HansenL18,TCS:Hansen19}
various refinements of Nash equilibria~\cite{AM:Nash51}. The
motivation for introducing these refinements has been to eliminate
undesirable equilibria, e.g., those relying on playing dominated
strategies.

The quasi-proper equilibrium, introduced by van Damme~\cite{IJGT:vanDamme84},
is one of the more refined solution concepts for extensive form games. Any
quasi-proper equilibrium is quasi-perfect, and therefore also sequential, and
also trembling hand perfect in the associated normal form game. Beyond being a
further refinement of the quasi-perfect equilibrium~\cite{IJGT:vanDamme84}, it
is also conceptually related in that it addresses a deficiency of the direct
translation of a normal form solution concept to extensive form games. One of
the most well known refinements is Selten's trembling hand perfect equilibrium,
originally defined~\cite{Selten65} for normal form games, and the solution
concept is usually referred to as normal-form perfect. This can be translated
to extensive form games, by applying the trembling hand definition to each
information set of each player, which yields what is now known as
extensive-form perfect equilibria~\cite{IJGT:Selten75}. However, this
translation introduces undesirable properties, first pointed out by
Mertens~\cite{GEB:Mertens95}. Specifically, Mertens presents a certain
two-player voting game where all extensive-form perfect equilibria have weakly
dominated strategies in their support. That is, extensive-form perfection is in
general inconsistent with {\em admissibility}. Mertens argues that
quasi-perfection is conceptually superior to Selten's notion of extensive-form
perfection, as it avoids the cause of the problem in Mertens' example. It
achieves this with a subtle modification of the definition of extensive-form
perfect equilibria, which in effect allows each player to plan as if they
themselves were unable to make any future mistakes. Further discussion of
quasi-perfection can be found in the survey of Hillas and
Kohlberg~\cite{HGT:HillasK02}.

One of the most restrictive equilibrium refinements of normal-form games is
that of Myerson's normal-form proper equilibrum~\cite{IJGT:Myerson78}, which is
a refinement of Selten's normal-form perfect equilibrium. Myerson's definition
can similarly be translated to extensive form, again by applying the definition
to each information set of each player, which yields the extensive-form proper
equilibria. Not surprisingly, all extensive-form proper equilibria are also
extensive-form perfect. Unfortunately, this also means that Merten's critique
applies equally well to extensive-form proper equilibria. Again, the definition
can be subtely modified to sidestep Merten's example, which then gives the
definition for quasi-proper-equilibria~\cite{IJGT:vanDamme84}. It is exactly
this solution concept that is the focus of this paper.

\subsection{Contributions}

The main novel idea of the paper is a new perturbation of the strategy space of
an extensive form game of perfect recall, in which a Nash equilibrium is an
$\eps$-quasi-proper equilibrium of the original game. This construction works
for \emph{any} number of players and in particular directly gives a new proof of
existence of quasi-proper equilibria for general $n$-player games.

From a computational perspective we can, in the important case of two-player
games, exploit the new pertubation in conjunction with the sequence form of
extensive form games~\cite{GEB:Koller96} to compute a symbolic quasi-proper
equilibrium by solving a Linear Complementarity Problem. This immediately
implies $\PPAD$-membership for the task of computing a symbolic quasi-proper
equilibrium. For the case of zero-sum games a quasi-proper equilibrium can be
computed by solving just a Linear program which in turn gives a polynomial time
algorithm.

For games with more than two players there is, from the viewpoint of
computational complexity, no particular advantage in working with the sequence
form. Instead we work directly with behavior strategies and go via so-called
$\delta$-almost $\eps$-quasi-proper equilibrium, which is a relaxation of
$\eps$-quasi-proper equilibrium.  We show $\FIXP_a$-membership for the task of
computing an approximation of a quasi-proper equilibrium. We leave the question
of $\FIXP$-membership as an open problem similarly to previous results about
computing Nash equilibrium refinements in games with more than two
players~\cite{SAGT:EtessamiHMS14,GEB:Etessami2021,EC:HansenL18}.

Since we work with refinements of Nash equilibrium, $\PPAD$-hardness for
two-player games and $\FIXP_a$-hardness for $n$-player games, with $n\geq 3$
follow directly. This combined with our membership results for $\PPAD$ and
$\FIXP_a$ implies $\PPAD$-completeness and $\FIXP_a$-completeness,
respectively.

\subsection{Relation to previous work}

Any strategic form game may be written as an extensive form game of comparable
size, and any quasi-proper equilibrium of the extensive form representation is a
proper equilibrium of the strategic form game. Hence our results fully
generalize previous results for computing~\cite{EC:Sorensen12} or
approximating~\cite{EC:HansenL18} a proper equilibrium. The generalization is
surprisingly clean, in the sense that if a bimatrix game is translated into an
extensive form game, the strategy constraints introduced in this paper will end
up being identical to those defined in~\cite{EC:Sorensen12} for the given
bimatrix game. This is surprising, since a lot of details have to align for this
structure to survive through a translation to a different game model.  Likewise,
if a strategic form game with more than two players is translated into an
extensive form game, the fixed point problem we construct in this paper is
identical to that for strategic form games~\cite{EC:HansenL18}.

The quasi-proper equilibria are a subset of the quasi-perfect
equilibria, so our positive computational results also generalize the
previous results for quasi-perfect
equilibria~\cite{ET:MiltersenSorensen10}. Again, the generalization is
clean; if all choices in the game are binary, then quasi-perfect and
quasi-proper coincide, and the constraints introduced in this paper
work out to be exactly the same as those for computing a quasi-perfect
equilibrium. The present paper thus manages to cleanly generalize two
different constructions in two different game models.

\section{Preliminaries}

\subsection{Extensive Form Games}
A game $\Gamma$ in \emph{extensive form} of imperfect information
with~$n$ players is given as follows. The structure of $\Gamma$ is
determined by a finite tree $T$. For a non-leaf node $v$, let $S(v)$
denote the set of immediate successor nodes. Let $Z$ denote the set of
leaf nodes of $T$. In a leaf-node $z \in Z$, player $i$ receives
utility $u_i(z)$. Non-leaf nodes are either chance-nodes or
decision-nodes belonging to one of the players. To every chance
node~$v$ is associated a probability distribution on $S(v)$. The set
$P_i$ of decision-nodes for Player~$i$ is partitioned into information
sets. Let $H_i$ denote the information sets of Player~$i$. To every
decision node $v$ is associated a set of $\Abs{S(v)}$ actions and
these label the edges between $v$ and $S(v)$. Every decision node
belonging to a given information set $h$ shares the same set $C_h$ of
actions. Define $m_h=\abs{C_h}$ to be the number of actions of every
decision node of $h$. The game $\Gamma$ is of \emph{perfect recall} if
every node $v$ belonging to an information set $h$ of Player~$i$ share
the same sequence of actions and information sets of Player~$i$ that
are observed on the path from the root of $T$ to $v$.  We shall only
be concerned with games of perfect recall~\cite{AMS:Kuhn53}.

A local strategy for Player~$i$ at information set $h \in H_i$ is a probability
distribution $b_{ih}$ on $C_h$ assigning a behavior probability to each action
in $C_h$ and in turn induces a probability distribution on $S(v)$ for every
$v \in h$. A local strategy $b_{ih}$ for every information set $h \in H_i$
defines a behavior strategy $b_i$ for Player~$i$. The behavior strategy $b_i$ is
fully mixed if $b_{ih}(c)>0$ for every $h \in H_i$ and every $c \in C_h$. Given
a local strategy $b'_{ih}$ denote by $b_i/b'_{ih}$ the result of replacing
$b_{ih}$ by $b'_{ih}$. In particular if $c \in C_h$ we let $b_i/c$ prescribe
action $c$ with probability~$1$ in $h$. For another behavior strategy $b'_i$ and
an information set $h$ for Player~$i$ we let $b_i\subslash{h}b'_i$ denote the
behavior strategy that chooses actions according to $b_i$ until $h$ is reached
after which actions are chosen according to $b'_i$. We shall also write
$b_i\subslash{h}b'_i/c=b_i\subslash{h}(b'_i/c)$. A behavior strategy profile
$b=(b_1,\dots,b_n)$ consists of a behavior strategy for each player. Let $B$ be
the set of all behavior strategy profiles of $\Gamma$. We let
$b_{-i}=(b_1,\dots,b_{i-1},b_{i+1},\dots,b_n)$ and
$(b_{-i};b'_i)=b/b'_i = (b_1,\dots,b_{i-1},b'_i,b_{i+1},\dots,b_n)$. Furthermore, for
simplicity of notation, we define $b\subslash{h}b'_i = b/(b_i\subslash{h}b'_i)$,
and $b\subslash{h}b'_i/c=b/(b_i\subslash{h}b'_i/c)$.

A behavior strategy profile $b=(b_1,\dots,b_n)$ gives together with
the probability distributions of chance-nodes a probability
distribution on the set of paths from the root-node to a leaf-node of
$T$. We let $\rho_b(v)$ be the probability that $v$ is reached by this
path and for an information set $h$ we let
$\rho_b(h)=\sum_{v \in h}\rho_b(v)$ be the total probability of
reaching a node of $h$. Note that we define $\rho_b(v)$ for all nodes
$v$ of $T$. When $\rho_b(h)>0$ we let $\rho_b(v \mid h)$ be the
conditional probability that node $v$ is reached given that $h$ is
reached. The \emph{realization weight} $\rho_{b_i}(h)$ for Player~$i$
of an information set $h\in H_i$ is the product of behavior
probabilities given by $b_i$ on any path from the root to $h$. Note
that this is well-defined due to the assumption of perfect recall.

Given a behavior strategy profile $b=(b_1,\dots,b_n)$, the payoff to
Player~$i$ is $U_i(b)=\sum_{z\in Z} u_i(z)\rho_b(z)$. When
$\rho_b(h)>0$ the conditional payoff to Player~$i$ given that $h$ is
reached is then $U_{ih}(b)=\sum_{z\in Z}u_i(z)\rho_b(z\mid h)$.

Realization weights are also defined on actions, to correspond to Player~$i$'s
weight assigned to the given action:

\begin{equation}
  \forall h\in H_i, c\in C_h:\quad \rho_{b_i}(c) = \rho_{b_i}(h)b_i(c)
\end{equation}

We note that the realization weight of an information set is equal to
that of the most recent action by the same player, or is equal to~1,
if no such action exists.

A realization plan for Player~$i$ is a strategy specified by its
realization weights for that player. As shown by
Koller~et~al.~\cite{GEB:Koller96}, the set of valid realization
weights for Player~$i$ can be expressed by the following set of linear
constraints
\begin{equation}\label{eq:seqform}
  \forall h\in H_i:\quad \rho_{b_i}(h)=\sum_{c\in
    C_h}\rho_{b_i}(c)\quad\land\quad\forall c\in C_h:
  \rho_{b_i}(c)\geq0
\end{equation}
in the variables $\rho_{b_i}(c)$ letting $\rho_{b_i}(h)$ refer to the
realization weight of the most recent action of Player~$i$ before
reaching information set $h$ or to the constant~1 if $h$ is the first
time Player~$i$ moves.  This formulation is known as the sequence
form~\cite{GEB:Koller96}, and has the advantage that for two-player
games, the utility of each player is bilinear, i.e., linear in the
realization weights of each player. As shown by Koller~et~al.\ this
allows the equilibria to be characterized by the solutions to a Linear
Complementarity Problem for general sum games, and as solutions to a
Linear Program for zero-sum games. We will build on this insight for
computing quasi-proper equilibria of two-player games.

Given a behavior strategy for a player, the corresponding realization
plan can be derived by multiplying the behavior probability of an action
with the realization weight of its information set. However, it is not
always the case that the reverse is possible. The behavior probability
of an action is the ratio of the realization weight of an action to the
realization weight of its information set, but if any of the
preceeding actions by the player have probability 0, the ratio works out
to $\frac00$. In the present paper, the restriction on the strategy
space ensures that no realization weight is zero, until we have
retrieved the behavior probabilities.

A strategy profile $b$ is a Nash equilibrium if for every $i$ and
every behavior strategy profile $b'_i$ of Player~$i$ we have
$U_i(b)\geq U_i(b/b'_i)$. Our object of study is quasi-proper
equilibrium defined by van~Damme~\cite{IJGT:vanDamme84} refining the
Nash equilibrium. We first introduce a convenient notation for
quantities used in the definition. Let $b$ be a behavior strategy
profile, $h$ an information set of Player~$i$ such that $\rho_b(h)>0$,
and $c \in C_h$. We then define
\begin{equation}
  K_i^{h,c}(b) = \max_{b'_i} U_{ih}(b\subslash{h}b'_i/c) \enspace .
\end{equation}
When $b'_i$ is a pure behavior strategy we say that
$b\subslash{h}b'_i$ is a $h$-local \emph{purification} of $b$. We note
that $U_{ih}(b\subslash{h}b'_i/c)$ always assumes its maximum for a
pure behavior strategy $b'_i$.
\begin{definition}[Quasi-proper equilibrium]
  Given $\eps>0$, a behavior strategy profile~$b$ is an
  $\eps$-quasi-proper equilibrium if $b$ is fully mixed and satisfies
  for every~$i$, every information set~$h$ of Player~$i$, and
  every $c,c' \in C_h$, that $b_{ih}(c) \leq \eps b_{ih}(c')$ whenever
  $K_i^{h,c}(b) < K_i^{h,c'}(b)$.
  
  A behavior strategy profile $b$ is a quasi-proper equilibrium if and
  only if it is a limit point of a sequence of $\eps$-quasi-proper
  equilibria with $\eps \rightarrow^+0$.
\end{definition}
We shall also consider a relaxation of quasi-proper equilibrium
in analogy to relaxations of other equilibrium refinements due to
Etessami~\cite{GEB:Etessami2021}.
\begin{definition}
  Given $\eps>0$ and $\delta>0$, a behavior strategy profile~$b$ is a
  $\delta$-almost $\eps$-quasi-proper equilibrium if $b$ is fully
  mixed and satisfies for every Player~$i$, every information set~$h$
  of Player~$i$, and every $c,c' \in C_h$ that
  $b_{ih}(c) \leq \eps b_{ih}(c')$ whenever
  $K_i^{h,c}(b) + \delta \leq K_i^{h,c'}(b)$.
\end{definition}

\subsection{Strategic Form Games}
A game $\Gamma$ in \emph{strategic form} with $n$ players is given as
follows. Player~$i$ has a set $S_i$ of \emph{pure strategies}. To a
pure strategy profile $a=(a_1,\dots,a_n)$ Player~$i$ is given utility
$u_i(a)$. A mixed strategy $x_i$ for Player~$i$ is a probability
distribution on $S_i$. We identify a pure strategy with the mixed
strategy that selects the pure strategy with probability~$1$. A
strategy profile $x=(x_1,\dots,x_n)$ consists of a mixed strategy for
each player. To a strategy profile $x$ Player~$i$ is given utility
$U_i(x)=\sum_{a \sim x} u_i(a)\prod_j x_j(a_j)$. A strategy
profile~$x$ is fully mixed if $x_i(a_i)>0$ for all~$i$ and all
$a_i\in S_i$. We let
$x_{-i}=(x_1,\dots,x_{i-1},x_{i+1},\dots,x_n)$. Given a
strategy~$x'_i$ for Player~$i$ we define
$(x_{-i};x'_i)=x/x'_i=(x_1,\dots,x_{i-1},x'_i,x_{i+1},\dots,x_n)$.

A strategy profile $x$ is a Nash equilibrium if for every $i$ and
every strategy~$x'_i$ of Player~$i$ we have $U_i(x/x'_i)\leq
U_i(x)$. Myerson defined the notion of proper
equilibrium~\cite{IJGT:Myerson78} refining the Nash equilibrium.
\begin{definition}[Proper equilibrium]
  Given $\eps>0$, a strategy profile $x$ is an $\eps$-proper
  equilibrium if $x$ is fully mixed and satisfies for every~$i$ and
  every $c,c' \in S_i$ that $x_i(c) \leq \eps x_i(c')$ whenever
  $U_i(x_{-i};c) < U_i(x_{-i};c')$.

  A strategy profile $x$ is a proper equilibrium if and only if it is
  a limit point of a sequence of $\eps$-proper equilibria with
  $\eps\rightarrow^+0$.
\end{definition}
For proper equilibrium we also consider a relaxation as suggested by
Etessami~\cite{GEB:Etessami2021}.
\begin{definition}
  Given $\eps>0$ and $\delta>0$, a strategy profile $x$ is a
  $\delta$-almost $\eps$-proper equilibrium if $x$ is fully
  mixed and satisfies for every~$i$ and
  every $c,c' \in S_i$ that $x_i(c) \leq \eps x_i(c')$ whenever
  $U_i(x_{-i};c) +\delta \leq  U_i(x_{-i};c')$.
\end{definition}

\subsection{Complexity Classes}
We give here only a brief description of the classes $\PPAD$ and
$\FIXP$ and refer to Papadimitriou~\cite{JCSS:Papadimitriou94} and
Etessami and Yannakakis~\cite{SICOMP:EtessamiY10} for detailed
definitions and discussion of the two classes.

$\PPAD$ is a class of discrete total search problems, whose totality is
guaranteed based on a parity argument on a directed graph. More formally $\PPAD$
is defined by a canonical complete problem \textsc{EndOfTheLine}. Here a
directed graph is given implicitly by predecessor and successor circuits, and
the search problem is to find a degree~1 node different from a given degree~1
node. We do not make direct use of the definition of $\PPAD$ but instead prove
$\PPAD$-membership indirectly via Lemke's algorithm~\cite{MS:Lemke65} for
solving a Linear Complementarity Problem (LCP).

$\FIXP$ is the class of real-valued total search problems that can be
cast as Brouwer fixed points of functions represented by
$\{+,-,*,/,\max,\min\}$-circuits computing a function mapping a convex
polytope described by a set of linear inequalities to itself. The
class $\FIXP_a$ is the class of \emph{discrete} total search problems
that reduce in polynomial time to \emph{approximate} Brouwer fixed
points. We will prove $\FIXP_a$ membership directly by constructing an
appropriate circuit.

\section{Two-Player Games}

In this section, we prove that computing a single quasi-proper equilibrium of a
two-player game $\Gamma$ can be done in \PPAD{}, and in the case of zero-sum
games, it can be computed in \P{}. We are using the same overall approach as has
been used for computing quasi-perfect equilibria of extensive
form~games~\cite{ET:MiltersenSorensen10}, proper equilibria of
two-player~games~\cite{EC:Sorensen12}, and proper equilibria of poly-matrix
games~\cite{EC:HansenL18}.

The main idea is to construct a new game $\Gamma_\eps$, where the
strategy space is slightly restricted for both players, in such a way
that equilibria of the new game are \eps-quasi-proper equilibria of the
original game. This construction also provides a new proof of existence
for quasi-proper equilibria of $n$-player games, since there is nothing
in neither the construction nor the proof that requires the game to have
only two players. However, for two players, the strategy constraints can
be enforced using a symbolic infinitesimal \eps, which can be part of
the solution output, thereby providing a witness of the quasi-properness
of the computed strategy.

We will first describe the strategy constraints. At a glance, the construction
consists of fitting the strategy constraints for \eps-proper
equilibria~\cite{EC:Sorensen12} into the strategy constraints of each of the
information sets of the sequence form~\cite{GEB:Koller96}, discussed in the
preliminaries section, equation (\ref{eq:seqform}).

The constraints for \eps-proper equilibria~\cite{EC:Sorensen12} restricts the
strategy space of each player to be an \eps-permutahedron. Before the technical
description of this, we define the necessary generalization of the
permutahedron. A permutahedron is traditionally over the vector $(1,\ldots,n)$,
but it generalizes directly to any other set as well.

\begin{definition}[Permutahedron]
  Let $\alpha \in \RR^m$ with all coordinates being distinct. A
  permutation $\pi \in S_m$ acts on $\alpha$ by permuting the
  coordinates of $\alpha$, i.e.~$(\pi(\alpha))_i=\alpha_{\pi(i)}$. We
  define the permutahedron $\perm{\alpha}$ over $\alpha$ to be the
  convex hull of the set $\{\pi(\alpha) \mid \pi \in S_m\}$ of the
  $m!$ permutations of the coordinates of $\alpha$.
\end{definition}
A very useful description of the permutahedron is by its $2^m-2$
facets.
\begin{proposition}[Rado~\cite{JLMS:Rado52}]
  Suppose $\alpha_1>\alpha_2>\cdots>\alpha_m$. Then
  \[
    \perm{\alpha}=  \biggl\{x \in \RR^m \biggm| \sum_{i=1}^m x_i = \sum_{i=1}^m \alpha_i \wedge \forall  S\notin \{\emptyset,[m]\} :  \sum_{c \in S} x_c \geq \sum_{i=1}^{\abs{S}}\alpha_{m-i+1} \biggr\} \enspace .
  \]\label{PROP:PermFacets}
\end{proposition}

As each inequality of Proposition~\ref{PROP:PermFacets} define a facet of the
permutahedron, any direct formulation of the permutahedron over $n$ elements
requires $2^n-2$ inequalities. Goemans~\cite{MP:Goemans2015} gave an
asymptotically optimal extended formulation for the permutahedron, using
$O(n\log n)$ additional constraints and variables. This allows a compact
representation, which allows us to use \eps-permutahedra~\cite{EC:Sorensen12}
as building blocks for our strategy constraints.

The \eps-permutahedron defined in~\cite{EC:Sorensen12} is a
permutahedron over the vector $(1,\eps,\eps^2,\ldots,\eps^{m-1})$,
normalized to sum to~1. We need to generalize this, so that it can sum
to any value $\rho$, and in a way that does not require normalization.
In the following, we will abuse notation slightly, and use $\rho$
without subscript as a real number, since it will shortly be replaced
by a realization weight for each specific information set.

\begin{definition}[$\eps$-Permutahedron]
  For real $\rho>0$, integers $k \geq 0$ and $m\geq 1$, and $\eps>0$
  such that $\rho\geq \eps^k$, define the vector
  $p_\eps(\rho,k,m) \in \RR^m$ by
  \[
    (p_\eps(\rho,k,m))_i = \begin{cases} \rho - (\eps^{k+1}+\dots+\eps^{k+m-1}) & , i=1\\ \eps^{k+i-1} & , i>1  \end{cases}
    \enspace ,
  \]
  and define the $\eps$-permutahedron  $\Pi_\eps(\rho,k,m)=\perm{p_\eps(\rho,k,m)} \subseteq \RR^m$.
\end{definition}
We shall be viewing $\eps$ as a variable. Note that, by definition,
$\norm{p_\eps(\rho,k,m)}_1=\rho$.

\begin{lemma}
  Assume $0<\eps\leq 1/3$ and $\rho \geq \eps^k$, for a given integer~$k\geq 0$. Then for every $1 \leq i < m$ we have
  $(p_\eps(\rho,k,m))_i \geq (p_\eps(\rho,k,m))_{i+1} / (2\eps)$.
  \label{LEM:p_eps_ratio}
\end{lemma}
\begin{proof}
  The statement clearly holds for $i>1$. Next we see that 
  $(p_\eps(\rho,k,m))_1=\rho-\eps^{k+1}(1-\eps^{m-1})/(1-\eps) \geq \eps^k -\eps^{k+1}/(1-\eps) = (1/\eps-1/(1-\eps))\eps^{k+1}\geq \eps^{k+1}/(2\eps)=(p_\eps(\rho,k,m))_{2} / (2\eps)$.
\end{proof}

We are now ready to define the perturbed game $\Gamma_\eps$.
\begin{definition}[Strategy constraints]\label{DEF:StratCons}
  For each player $i$, and each information set $h\in H_i$, let
  $k_h=\sum_{h'<h} m_{h'}$ be the sum of the sizes of the action sets
  at information sets visited by Player~$i$ before reaching
  information set $h$. Now, in the perturbed game $\Gamma_\eps$,
  restrict
  $(\rho_{b_i}(c_1),\rho_{b_i}(c_2),\ldots,\rho_{b_i}(c_{m_h}))$ to be
  in $\Pi_\eps(\rho_{b_i}(h),k_h,m_h)$.
\end{definition}
Notice that the strategy constraints for the first information set a
player visits is identical to the strategy constraints for proper
equilibria of bimatrix games.

The next three lemmas describe several ways we may modify coordinates of points of
$\Pi_\eps(\rho,k,m)$ while staying within $\Pi_\eps(\rho',k,m)$ for appropriate
$\rho'$. These are needed for the proof of our main technical result,
Proposition~\ref{PROP:2P}, below.
\begin{lemma}
  Let $0<\eps<1/3$, $\rho \geq \eps^k$, and
  $x \in \Pi_\eps(\rho,k,m)$. Suppose for distinct $c$ and $c'$ we
  have $x_c > 2\eps x_{c'}$. Then there exists $\delta>0$ such that
  $x+\delta (e_{c'}-e_c) \in \Pi_\eps(\rho,k,m)$ (here as usual $e_i$
  denotes the $i$-unit vector).
  \label{LEM:ModPerm1}
\end{lemma}
\begin{proof}
  By definition of $\Pi_\eps(\rho,k,m)$ we may write $x$ as a convex
  combination of the corner points of $\Pi_\eps(\rho,k,m)$,
  $x = \sum_{\pi \in S_m} w_\pi \pi(p_\eps(\rho,k,m))$, where
  $w_\pi\geq 0$ and $\sum_{\pi \in S_m} w_\pi = 1$. There must exist a
  permutation $\pi$ such that $w_\pi>0$ and
  $\pi^{-1}(c)<\pi^{-1}(c')$, since otherwise $x_c \leq 2\eps x_{c'}$
  by Lemma~\ref{LEM:p_eps_ratio}. Let $\pi' \in S_m$ such that
  $\pi'(\pi^{-1}(c))=c'$, $\pi'(\pi^{-1}(c'))=c$, and $\pi'(i)=\pi(i)$
  when $\pi(i)\notin \{c,c'\}$. We then have that
  \[
    x'=x+w_\pi (\pi'(p_\eps(\rho,k,m)) - \pi(p_\eps(\rho,k,m))) \in
    \Pi_\eps(\rho,k,m) \enspace .
  \]
  Note now that
  $\pi'(p_\eps(\rho,k,m)) - \pi(p_\eps(\rho,k,m))$ is equal to
  \[
    ((p_\eps(\rho,k,m))_{\pi^{-1}(c)} - (p_\eps(\rho,k,m))_{\pi^{-1}(c')})(e_{c'}-e_c) \enspace .
  \]
  Since
  $(p_\eps(\rho,k,m))_{\pi^{-1}(c)} > (p_\eps(\rho,k,m))_{\pi^{-1}(c')}$, the
  statement follows.  \end{proof}

\begin{lemma}
  Let $x \in \Pi_\eps(\rho,k,m)$ where $\rho\geq \eps^k$. Then
  $x+\delta e_c \in \Pi(\rho+\delta,k,m)$ for any $\delta>0$ and $c$.
  \label{LEM:ModPerm2}
\end{lemma}
\begin{proof}
  This follows immediately from Proposition~\ref{PROP:PermFacets}
  since  the inequalities defining the facets of
  $\Pi_\eps(\rho,k,m)$ and $\Pi_\eps(\rho+\delta,k,m)$ are exactly the
  same.
\end{proof}

\begin{lemma}
  Let $x \in \Pi_\eps(\rho,k,m)$ where $0<\eps\leq 1/2$ and
  $\rho > \max(\eps^k,2m\eps^{k+1})$. Let $c$ be such that
  $x_c\geq x_{c'}$ for all $c'$. Then
  $x-\delta e_c \in \Pi_\eps(\rho-\delta,k,m)$ for any
  $\delta \leq \min(\rho-\eps^k,\rho/m-2\eps^{k+1})$.
  \label{LEM:ModPerm3}
\end{lemma}
\begin{proof}
  Since $\delta \leq \rho - \eps^k$ we have $\rho-\delta \geq \eps^k$,
  thereby satisfying the definition of $\Pi_\eps(\rho-\delta,k,m)$.  By the
  choice of $c$ we have that $x_c\geq \rho/m$. Since we also have
  $\delta \leq \rho/m-2\eps^{k+1}$ it follows that
  $x_c - \delta \geq 2\eps^{k+1}$. Thus
  $x_c-\delta \geq \eps^{k+1}+\dots+\eps^{k+m-1}$. It then follows
  immediately from Proposition~\ref{PROP:PermFacets} that
  $x-\delta e_c \in \Pi_\eps(\rho-\delta,k,m)$, since any inequality
  given by $S$ with $c \in S$ is trivially satisfied, and any
  inequality with $c \notin S$ is unchanged from $\Pi_\eps(\rho,k,m)$.
\end{proof}
We are now in position to prove correctness of our approach.
\begin{proposition}\label{PROP:2P}
  Any Nash equilibrium of $\Gamma_\eps$ is a $2\eps$-quasi-proper
  equilibrium of $\Gamma$, for any sufficiently small $\eps>0$.
\end{proposition}
\begin{proof}
  Let $b$ be a Nash equilibrium of $\Gamma_\eps$. Consider 
  Player~$i$ for any~$i$, any information set $h \in H_i$, and let $c,c' \in h$ be
  such that $b_{ih}(c) > 2\eps b_{ih}(c')$. We are then to show that
  $K_i^{h,c}(b) \geq K_i^{h,c'}(b)$, when $\eps>0$ is sufficiently
  small. Let $b'_i$ be such that
  $U_{ih}(b\subslash{h}b'_i/c')=K_i^{h,c'}(b)$. We may assume that $b'_i$
  is a pure behavior strategy thereby making $b/b'_i$ a $h$-local
  purification. Let $H_{i,c'}$ be the set of those information sets of
  Player~$i$ that follow after $h$ when taking action $c'$ in
  $h$. Similarly, let $H_{i,c}$ be the set of those information sets
  of Player~$i$ that follow after taking action $c$ in $h$. Note that
  by perfect recall of $\Gamma$ we have that
  $H_{i,c'} \cap H_{i,c} = \emptyset$. Let $b^*_i$ be any pure
  behavior strategy of Player~$i$ choosing $c^*_h \in C_h$ maximizing
  $b_{ih}(c^*_h)$, for all $h \in H_i$. We claim that
  $U_{ih}(b\subslash{h}b^*_i/c)\geq K_i^{h,c'}(b)$ for all
  sufficiently small $\eps>0$.

  Let $x_i$ be the realization plan given by $b_i$, let $x'$ be the
  realization plan given by $b_i\subslash{h}b'_i/c'$, and let $x^*_i$
  be the realization plan given by $b_i\subslash{h}b^*_i/c$. We shall
  next apply Lemma~\ref{LEM:ModPerm1} to $h$, Lemma~\ref{LEM:ModPerm2}
  to all $h' \in H_{i,c'}$, and Lemma~\ref{LEM:ModPerm3} to all
  $h^* \in H_{i,c}$ assigned positive realization weight by
  $b_i\subslash{h}b^*_i/c$, to obtain that for all sufficiently small
  $\eps>0$ there is $\delta>0$ such that
  $\widetilde{x}_i=x_i+\delta(x'_i-x^*_i)$ is a valid realization plan
  of~$\Gamma_\eps$.

  Lemma~\ref{LEM:ModPerm2} can be applied whenever $\eps>0$ is
  sufficiently small, whereas Lemma~\ref{LEM:ModPerm1} in addition
  makes use of the assumption that $b_{ih}(c)>2\eps b_{ih}(c')$. To
  apply Lemma~\ref{LEM:ModPerm3}, we need to prove that the player's
  realization weight is sufficiently large for the relevant
  information sets, specifically $\rho_{h'}>\eps^{k_{h'}}$ for each
  relevant information set $h'$. Since $b^*_i$ is pure, Player $i$'s
  realization weight, $\rho_{h'}$ for each information set $h'$ in
  $H_{i,c}$ is either 0 or $\rho_c$. Since
  $b_{ih}(c) > 2\eps b_{ih}(c')$, we have that
  $\rho_c > \eps^{k_h+|C_h|-1}\geq\eps^{k_{h'}}$ as needed.
  
  Thus, consider $\eps>0$ and $\delta>0$ such that $\widetilde{x}_i$
  is a valid realization plan and let $\widetilde{b}_i$ be the
  corresponding behavior strategy. Since $b$ is a Nash equilibrium we
  have $U_i(b/\,\widetilde{b}_i)\leq U_i(b)$. But
  $U_i(b/\,\widetilde{b}_i) = U_i(b) +
  \delta(U_i(b\subslash{h}b'_i/c')-U_i(b\subslash{h}b^*_i/c))$. It
  follows that
  $\delta(U_i(b\subslash{h}b'_i/c')-U_i(b\subslash{h}b^*_i/c))\leq 0$,
  and since $\delta>0$ we have
  $U_i(b\subslash{h}b^*_i/c)\geq
  U_i(b\subslash{h}b'_i/c')$. Equivalently,
  $U_{ih}(b\subslash{h}b^*_i/c)\geq U_{ih}(b\subslash{h}b'_i/c')$,
  which was to be proved.  Since $i$ and $h \in H_i$ were arbitrary,
  it follows that $b$ is a $2\eps$-quasi-proper equilibrium in
  $\Gamma$, for any sufficient small $\eps>0$.  \end{proof}

\begin{theorem}
  A symbolic \eps-quasi-proper equilibrium for a given two-player
  extensive form game with perfect recall can be computed by applying
  Lemke's algorithm to an LCP of polynomial size, and can be computed
  in \PPAD{}.
\end{theorem}

\begin{proof}
  Given an extensive form game $\Gamma$, construct the game
  $\Gamma_\eps$. The strategy constraints
  (Definition~\ref{DEF:StratCons}) are all expressed directly in terms
  of the realization weights of each player. Using
  Goemans'~\cite{MP:Goemans2015} extended formulation, the strategy
  constraints require only $O(\sum_h |C_h|\log|C_h|)$ additional
  constraints and variables, which is linearithmic in the size of the
  game. Furthermore, all occurrences of \eps{} are on the right-hand
  side of the linear constraints. These constraints fully replace the
  strategy constraints of the sequence form~\cite{GEB:Koller96}.
  In the sequence form, there is a single equality per information set,
  ensuring conservation of the realization weight. In our case, this
  conservation is ensured by the permutahedron constraint for each
  information set.

  In the case of two-player
  games, the equilibria can be captured by an LCP of polynomial size,
  which can be solved
  using Lemke's algorithm~\cite{MS:Lemke65}, if the strategy constraints
  are sufficiently well behaved. Since the added strategy constraints
  is a collection of constraints derived from Goemans' extended
  formulation, the proof that the constraints
  are well behaved is identical to the proofs of~\cite[Theorem 5.1 and
  5.4]{EC:Sorensen12}, which we will therefore
  omit here. Following the approach of \cite{ET:MiltersenSorensen10}
  the solution to the LCP can be made to contain the symbolic \eps{},
  with the probabilities of the strategies being formal polynomials in
  the variable \eps{}.

  By Proposition~\ref{PROP:2P}, equilibria of $\Gamma_\eps$ are
  \eps-quasi-proper equilibria of $\Gamma$. All realization weights of the
  computed realization plans are formal polynomials in~\eps{}. Finally, from
  this we may express the $\eps$-quasi-proper equilibrium in behavior
  strategies, where all probabilities are rational functions in $\eps$.
  \end{proof}

Having computed a symbolic $\eps$-quasi-proper equilibrium for $\Gamma$ it is
easy to compute the limit for $\eps{}\rightarrow0$, thereby giving a
quasi-proper equilibrium of $\Gamma$. It is crucial here that we first convert
into behavior strategies before computing the limit. In the case of zero-sum
games, the same construction can be used to construct a linear program of
polynomial size, whose solution would provide quasi-proper equilibria of the
given game. This is again analogous to the approach of
\cite{ET:MiltersenSorensen10} and further details are hence omitted.
\begin{theorem}
  A symbolic \eps-quasi-proper equilibrium for a given two-player extensive form
  zero-sum game with perfect recall can be computed in polynomial time.
\end{theorem}

\section{Multi-Player Games}
In this section we show that approximating a quasi-proper equilibrium for a
finite extensive-form game $\Gamma$ with $n\geq 3$ players is
$\FIXP_a$-complete. As for two-player games, by Proposition~\ref{PROP:2P} an
$\eps$-quasi-proper equilibrium for $\Gamma$ could be obtained by computing an
equilibrium of the perturbed game $\Gamma_\eps$. But for more than two players
we do not know how to make efficient use of this connection. Indeed, from the
viewpoint of computational complexity there is no advantage in doing so. Our
construction instead works by directly combining the approach and ideas of the
proof of $\FIXP_a$-completeness for quasi-perfect equilibrium in extensive form
games by Etessami~\cite{GEB:Etessami2021} and of the proof of
$\FIXP_a$-completeness for proper equilibrium in strategic form games by Hansen
and Lund~\cite{EC:HansenL18}. We explain below how these are modified and
combined to obtain the result. The approach obtains $\FIXP_a$ membership,
leaving $\FIXP$-membership as an open problem. A quasi-proper equilibrium is
defined as a limit point of a sequence of $\eps$-quasi-proper equilibria, whose
existence was obtained by the Kakutani fixed point theorem by
Myerson~\cite{IJGT:Myerson78}. This limit point operation in itself poses a
challenge for $\FIXP$ membership. The use of the Kakutani fixed point theorem
presents a further challenge.
However, as we show below analougous to the case of proper
equilibria~\cite{EC:HansenL18}, these may be approximated by $\delta$-almost
$\eps$-quasi-proper equilibria, which in turn can be expressed as a set of
Brouwer fixed points. In fact we show that the corresponding search problem is
in $\FIXP$.

To see how to adapt the result of Hansen and Lund~\cite{EC:HansenL18} for
strategic form games to the setting of extensive form games, it is helpful to
compare the definitions of $\eps$-proper equilibrium and $\delta$-almost
$\eps$-proper equilibrium in strategic form games to the corresponding
definitions of $\eps$-quasi-proper equilibrium and $\delta$-almost $\eps$-proper
equilibrium in extensive form games.

In a strategic form game, Player~$i$ is concerned with the payoffs
$U_i(x_{-i},c)$, which we may think of as \emph{valuations} of all pure
strategies $c \in S_i$. The relationship between these valuations in turn place
constraints on the strategy $x_i$ chosen by Player~$i$ in an $\eps$-proper
equilibrium or a $\delta$-almost $\eps$-proper equilibrium. In an extensive form
game, Player~$i$ is in a given information set $h$ considering the payoffs
$K_i^{h,c}$, which we may similarly think of as \emph{valuations} of all actions
$c \in C_h$. The relationship between these valuations place constraints on the
local strategy $b_{ih}$ chosen by Player~$i$ in a $\eps$-quasi-proper
equilibrium or a $\delta$-almost $\eps$-proper equilibrium. These constraints
are completely analogous to those placed on the strategies in strategic form
games. This fact will allow us to adapt the constructions of Hansen and Lund by
essentially just changing the way the valuations are
computed. Etessami~\cite{GEB:Etessami2021} observed that these may be computed
using dynamic programming and gave a construction of formulas computing them.
\begin{lemma}[cf.~{\cite[Lemma~7]{GEB:Etessami2021}}]
  Given an extensive form game of perfect recall~$\Gamma$, a
  player~$i$, an information set $h$ of Player~$i$, and $c \in C_h$
  there is a polynomial size $\{+,-,*,/,\max\}$-formula $V_i^{h,c}$
  computable in polynomial time satisfying that for any fully mixed
  behavior strategy profile $b$ it holds that  $V_i^{h,c}(b)=K_i^{h,c}(b)$.
\label{LEM:K-formula}
\end{lemma}
We now state our result for multi-player games.
\begin{theorem}
  \label{THM:MainMultiPlayer}
  Given as input a finite extensive form game of perfect recall
  $\Gamma$ with $n$ players and a rational $\gamma>0$, the problem of
  computing a behavior strategy profile $b'$ such that there is a
  quasi-proper equilibrium $b$ of $\Gamma$ with
  $\norm{b'-b}_\infty < \gamma$ is $\FIXP_a$-complete.
\end{theorem}

Before presenting the proof of Theorem~\ref{THM:MainMultiPlayer} we
describe the changes needed to adapt the results of Hansen and
Lund~\cite{EC:HansenL18} to extensive-form games in more details.

The first step of the construction is to establish that to compute an
approximation to a quasi-proper equilibrium it is sufficient to
compute (an approximation to) an $\eps$-quasi-proper equilibrium, for
a sufficiently small $\eps>0$, and further to compute an approximation
to an $\eps$-quasi-proper equilibrium it is sufficient to compute a
$\delta$-almost $\eps$-quasi-proper equilibrium, for a sufficiently
small $\delta>0$. Both statements are obtained by invoking the
``almost implies near'' paradigm of
Anderson~\cite{TAMS:Anderson86}. The first statement generalizes
essentially verbatim from the case of proper equilibrium in strategic
form games~\cite[Lemma~4.2]{EC:HansenL18} and we omit the proof.
\begin{lemma}
  For any fixed extensive form game of perfect recall $\Gamma$, and
  any $\gamma>0$, there is an $\eps>0$ so that any
  $\eps$-quasi-proper equilibrium of $\Gamma$ has
  $\ell_\infty$-distance at most $\gamma$ to some quasi-proper
  equilibrium of $\Gamma$.
\label{LEM:Almost-implies-near1}
\end{lemma}
We now define a \emph{perturbed} version of $\Gamma$, restricting the
domain of local behavior strategies. For $\eps>0$ and a positive
integer $m$ define $\eta_m(\eps)=\eps^m/m$. The $\eta$-perturbed game
$\Gamma_\eta$ restricts a local behavior strategy in every information
set $h$ to use behavior probabilities at least $\eta_{m_h}(\eps)$. Let
$B_\eta$ be the set of such restricted behavior strategy profiles of
$\Gamma$. The proof of the second statement following below very
closely follows that of~\cite[Lemma~4.3]{EC:HansenL18}. For
completeness we give the proof.
\begin{lemma}
  For any fixed extensive form game of perfect recall $\Gamma$, any
  $\eps>0$ and any $\gamma>0$, there is a $\delta>0$ so that any
  $\delta$-almost $\eps$-quasi-proper equilibrium of $\Gamma$ in $B_\eta$
  has $\ell_\infty$-distance at most $\gamma$ to some
  $\eps$-quasi-proper equilibrium of $\Gamma$ in $B_\eta$.
\label{LEM:Almost-implies-near2}
\end{lemma}
\begin{proof}
  Suppose to the contrary there is a game $\Gamma$, $\eps>0$, and
  $\gamma>0$ so that for all $\delta>0$ there is a $\delta$-almost
  $\eps$-quasi-proper equilibrium $b_\delta$ of $\Gamma$ in $B_\eta$
  so that there is no $\eps$-quasi-proper equilibrium in $B_\eta$ in a
  $\gamma$-neighborhood (with respect to the $\ell_\infty$ norm) of
  $b_\delta$. Consider the sequence $(b_{1/n})_{n\in \NN}$. Since this is
  a sequence in a compact space, by the Bolzano-Weierstrass Theorem it
  has a convergent subsequence $(b_{1/n_r})_{r_\in\NN}$. Let
  $b^* = \lim_{r\rightarrow\infty}b_{1/n_r}$. We now claim that $b^*$
  is an $\eps$-quasi-proper equilibrium, which will contradict the
  statement that there is no $\eps$-quasi-proper equilibrium in a
  $\gamma$-neighborhood of any of the behavior strategy profiles
  $b_{1/n}$.

  First, since $b_{1/n_r} \in B_\eta$ for all $n_r$ we also have
  $b^* \in B_\eta$. In particular, $b^*$ is fully mixed. The
  functions $K_i^{h,c}$ are well defined on $B_\eta$. Define $\nu > 0$
  by
\[
  \nu = \min_{i,h,c,c'} \left\{K_i^{h,c'}(b^*)-K_i^{h,c}(b^*)\mid
    K_i^{h,c}(b^*)<K_i^{h,c'}(b^*)\right\} \enspace .
\]
By continuity of the functions $K_i^{h,c}$ we have
$\lim_{r \rightarrow \infty} K_i^{h,c}(b_{1/n_r}) = K_i^{h,c}(b^*)$,
for all $i$, $h$, and $c$. Thus let $N$ be an integer such that
$\Abs{K_i^{h,c}(b_{1/n_r}) - K_i^{h,c}(b^*)} \leq \nu/3$
and such that $1/n_r \leq \nu/3$, for all $i$,$h$,$c$, and $r \geq
N$.

Consider now an information set $h$ of Player~$i$ and $c,c' \in C_h$
such that $K_i^{h,c}(b^*) < K_i^{h,c'}(b^*)$. By construction, for any
$r \geq N$ we also have
$K_i^{h,c}(b_{1/n_r}) + 1/{n_r} \leq K_i^{h,c'}(b_{1/n_r})$. Since
$b_{1/n_r}$ is a $(1/n_r)$-almost $\eps$-quasi-proper equilibrium it follows
that $(b_{1/n_r})_{ih}(c) \leq \eps (b_{1/n_r})_{ih}(c')$. Taking the
limit $r \rightarrow \infty$ we also have
$b^*_{ih}(c) \leq b^*_{ih}(c')$, which shows that $b^*$ is an
$\eps$-quasi-proper equilibrium.
\end{proof}

The second step is to show that given $\Gamma$, $\delta>0$, and $\eps>0$, the
task of computing a $\delta$-almost $\eps$-quasi-proper equilibrium of $\Gamma$
belongs to $\FIXP$.  We outline the details of this below, using a slightly
different notation compared to~\cite{EC:HansenL18}.
\begin{definition}[cf.~{\cite[Definition~4.4]{EC:HansenL18}}]
  Let $v \in \RR^m$, $x \in \RR_+^m$, $\delta>0$, and $\eps>0$. We say
  that $x$ satisfies the $\delta$-almost $\eps$-proper property with
  respect to valuation $v$ if and only if $x_c \leq \eps x_{c'}$
  whenever $v_c + \delta \leq v_{c'}$, for all $c,c'$.
\end{definition}

Hansen and Lund~\cite[Definition~4.6]{EC:HansenL18} define a function
$P_{m,\delta,\eps} : \RR_+^m \times \RR^m \rightarrow \RR_+^m$ as a
main ingredient of computing $\delta$-almost $\eps$-proper
equilibrium. This is given by 
\[
(P_{m,\delta,\eps}(x,v))_c = \min_{c'} \left\{\deltasel(x_{c},\eps x_{c'},v_{c'}-v_{c})\right\} \enspace ,
\]
where
\[
  \deltasel(x,y,z) =
  \begin{cases}
    x & \text{ if } z \leq 0\\ (1-z/\delta)x + (z/\delta) y & \text{
      if } 0 \leq z \leq \delta\\ y & \text{ if } \delta \leq z
  \end{cases}
\]
is the the $\delta$-approximate selection function, for $\delta>0$.

The function function $P_{m,\delta,\eps}$ then induces an operator
$P^v_{m,\delta,\eps} : \RR_+^m \rightarrow \RR_+^m$ by letting
$P^v_{m,\delta,\eps}(x) = P_{m,\delta,\eps}(x,v)$. Define
$\Delta_m=\{y \in \RR^m \mid \norm{y}_1=1; \forall j: y_j \geq 0\}$
and for $\eta>0$ define
$\Delta_m^\eta = \{y \in \RR^m \mid \norm{y}_1=1; \forall j: y_j \geq
\eta\}$. We may identify the points of $\Delta_m$ and $\Delta_m^\eta$
with probability distributions on a set of~$m$ elements. Let
$\tau_m \in \Delta_m$ be the uniform distribution on~$m$
elements. Note that $\tau_m \in \Delta_m^{1/m}$. We need the following
properties of $P^v_{m,\delta,\eps}$ proved by Hansen and Lund. We let
$(P^v_{m,\delta,\eps})^{\circ j}$ denote the $j$-th iteration of the
operator $P^v_{m,\delta,\eps}$.
\begin{lemma}[cf.~{\cite[Lemma~4.10]{EC:HansenL18}}]
  If $x \in \RR_+^m$ is a fixed point of $P^v_{m,\delta,\eps}$ then
  $x$ satisfies the $\delta$-almost $\eps$-proper property with respect
  to~$v$.
\end{lemma}
\begin{proposition}[cf.~{\cite[Lemma~4.11 and Proposition~4.15]{EC:HansenL18}}]
  Suppose $\eps \leq 1/m$. Then
  $(P^v_{m,\delta,\eps})^{\circ j}(\tau_m)$ in contained in
  $\Delta_m^{\eta_m}$ and satisfy the $\delta$-almost
  $\sqrt{\eps}$-proper property with respect to $v$ for all
  $j\geq 2m^2$.
  \label{PROP:UniformPIteration-is-almost-proper}
\end{proposition}
We now have everything needed for defining the fixed point problem. We
define a function
$F_{\eps,\delta} : B^{\eta(\eps^2)} \rightarrow B^{\eta(\eps^2)}$ as
follows. For $b \in B^{\eta(\eps^2)}$, define the following for every
$i$ and every information set $h$ of Player~$i$: First we let
$v_{ih} \in \RR_{m_h}$ be given by $(v_{ih})_c = V_i^{h,c}(b)$. We then
let $y_{ih}= (P^{v_ih}_{m_h,\delta,\eps})^{\circ 2m_h^2}(\tau_{m_h})$ and
$b'_{ih}=y_{ih}/\norm{y_{ih}}_1$. Finally define $F_{\eps,\delta}(b)=b'$.
\begin{proposition}
  Let $\delta>0$ and $0<\eps<1$. Then every fixed point
  $b \in B^{\eta(\eps^2)}$ of $F_{\eps,\delta}$ is a $\delta$-almost
  $\eps$-quasi-proper equilibrium of $\Gamma$.
\end{proposition}
\begin{proof}
  Suppose that $b \in B^{\eta(\eps^2)}$ is a fixed point of
  $F_{\eps,\delta}$. For every~$i$ and every information set $h$ of
  Player~$i$ follows that $b_{ih}=y_{ih}/\norm{y_{ih}}_1$. By
  Proposition~\ref{PROP:UniformPIteration-is-almost-proper} $y_{ih}$
  satisfies the $\delta$-almost $\eps$-proper property with respect to
  valuation $v_{ih}$. This implies that $b_{ih}$ satisfies the
  $\delta$-almost $\eps$-proper property with respect to valuation
  $v_{ih}$ as well. Since this holds for all~$i$ and~$h$, we can
  conclude that $b$ is a $\delta$-almost $\eps$-quasi-proper
  equilibrium.  \end{proof} By Lemma~\ref{LEM:K-formula} the
valuations $v_{ih}$ may be computed by a polynomial size
$\{+,-,*,/,\max\}$-formula. Likewise, as seen from their definition,
the functions $P_{m,\delta,\eps}$ may be computed by polynomial size
$\{+,-,*,/,\max,\min\}$-formulas. All these formulas may furthermore
be constructed in polynomial time. The function $F_{\eps,\delta}$ is
given by combining polynomially many such formulas into a
\emph{circuit}. In conclusion we obtain the following result,
analogously to \cite[Theorem~4.17]{EC:HansenL18}.
\begin{theorem}
  There exists a function
  $F_{\eps,\delta} : B^{\eta(\eps^2)} \rightarrow B^{\eta(\eps^2)}$
  that is given by a $\{+,-,*,/,\max,\min\}$-circuit computable in
  polynomial time from $\Gamma$, with the circuit having inputs $b$,
  $\eps>0$, and $\delta>0$, such that for all fixed $0<\eps<1$ and
  $\delta>0$, every fixed point of $F_{\eps,\delta}$ is a
  $\delta$-almost $\eps$-quasi-proper equilibrium of $\Gamma$. In
  particular, the problem of computing a $\delta$-almost
  $\eps$-quasi-proper equilibrium of a finite extensive form game of
  perfect recall $\Gamma$ is in $\FIXP$.
  \label{THM:almost-quasi-proper-FIXP}
\end{theorem}

The third step is to quantify how small $\eps>0$ and $\delta>0$ need
to be in order to guarantee that Lemma~\ref{LEM:Almost-implies-near1}
and Lemma~\ref{LEM:Almost-implies-near2} apply. Such bounds can be
obtained in a completely generic way using the general machinery of
real algebraic geometry, cf. Basu, Pollack, and
Roy~\cite{BasuPollackRoy2006}, and was applied for the same purpose in
previous
works~\cite{SAGT:EtessamiHMS14,GEB:Etessami2021,EC:HansenL18}. The
approach involves formalizing the statements of
Lemma~\ref{LEM:Almost-implies-near1} and
Lemma~\ref{LEM:Almost-implies-near2} in the first order theory of the
reals. More precisely, doing this for
Lemma~\ref{LEM:Almost-implies-near1} results in a formula depending on
$\Gamma$ and $\gamma$ with a free variable $\eps$. This formula is
built from the definition of a $\eps$-quasi-proper equilibrium as well
as the formula of Lemma~\ref{LEM:K-formula}. Applying quantifier
elimination to that formula and employing known bounds on the result
of this we obtain the following statement, analogously to \cite[Lemma~4.18]{EC:HansenL18}.
\begin{lemma}
  There exists a polynomial $q_1$ such that for any finite extensive
  form game $\Gamma$ of perfect recall and any $0<\gamma<1/2$,
  whenever $0<\eps<\gamma^{2^{q_1(\Abs{\Gamma})}}$
  any $\eps$-quasi-proper equilibrium of $\Gamma$ has
  $L_\infty$-distance at most $\gamma$ to some quasi-proper
  equilibrium of $\Gamma$.
  \label{LEM:infinitesimal1}
\end{lemma}
Similarly for Lemma~\ref{LEM:Almost-implies-near2} we construct a
formula depending on $\Gamma$, $\gamma$, and $\eps$ with a free
variable $\delta$. Again applying quantifier elimination to that
formula and employing known bounds on the result of this we obtain the
following statement, analogously to \cite[Lemma~4.19]{EC:HansenL18}.
\begin{lemma}
  There exists a polynomial $q_2$ such that for any finite extensive
  form game $\Gamma$ of perfect recall, any $0<\gamma<1/2$, and any
  $\eps>0$, whenever $0<\min(\delta,\eps)^{2^{q_2(\Abs{\Gamma})}}$ any
  $\delta$-almost $\eps$-quasi-proper equilibrium of $\Gamma$ has
  $L_\infty$-distance at most $\gamma$ to some $\eps$-quasi-proper
  equilibrium of $\Gamma$.
  \label{LEM:infinitesimal2}
\end{lemma}
We can now complete the proof of Theorem~\ref{THM:MainMultiPlayer}. As done for
the case of approximating proper equilibrium~\cite{EC:HansenL18}, the idea is to
construct two virtual infinitesimals $\delta \ll \eps$, given $\Gamma$ and
$\gamma>0$, by means of repeated squaring, according to
Lemma~\ref{LEM:infinitesimal1} and Lemma~\ref{LEM:infinitesimal2}.
\begin{proof}[Proof of Theorem~\ref{THM:MainMultiPlayer}]
  Given an extensive form game of perfect recall $\Gamma$ and a
  rational $\gamma>0$ we shall in polynomial time construct a
  $\{+,-,*,/,\max,\min\}$-circuit $C$ computing a function
  $F: B \rightarrow B$ such that any fixed point of $F$ is
  $\gamma$-close to a quasi-proper equilibrium of $\Gamma$. This is
  sufficient to establish $\FIXP_a$-membership.

  The circuit $C$ will first compute $\eps>0$ satisfying the condition
  of Lemma~\ref{LEM:infinitesimal1} by repeated squaring of $\gamma/2$
  exactly $q_1(\Abs{\Gamma})$ times. Then $C$ computes $\delta>0$
  satisfying the condition of Lemma~\ref{LEM:infinitesimal2} by
  repeated squaring of $\min(\gamma/2,\eps)$ exactly
  $q_2(\Abs{\Gamma})$ times. Next we need to restrict the input to
  $B^{\eta(\eps^2)}$ before we can apply the function
  $F_{\eps,\delta}$ of Theorem~\ref{THM:almost-quasi-proper-FIXP}. For
  this we need to map the input $x \in B$ into $B^{\eta(\eps^2)}$ by a
  mapping that is the identity function on $B^{\eta(\eps^2)}$. One way
  to achieve this (cf.~\cite{EC:HansenL18}) is to compute for every
  $i$ and $h$ a number $t_{ih}$ such that
  $\sum_{c \in C_h} \max(b_{ih}(c)-t_{ih},\eta_{m_h}(\eps^2))=1$ using
  a sorting network as done by Etessami and
  Yannakakis~\cite{SICOMP:EtessamiY10} and then map each $b_{ih}(c)$
  to $\max(b_{ih}(c)-t_{ih},\eta_{m_h}(\eps^2))$. Finally
  $F_{\eps,\delta}$ is applied to the output of this together with the
  constructed $\eps$ and $\delta$. By
  Theorem~\ref{THM:almost-quasi-proper-FIXP} any fixed-point of $F$ is
  then a $\delta$-almost $\eps$-quasi-proper equilibrium of
  $\Gamma$. By Lemma~\ref{LEM:infinitesimal1} this is $\gamma/2$-close
  to a $\eps$-quasi-proper equilibrium which in turn by
  Lemma~\ref{LEM:infinitesimal2} is $\gamma/2$-close to a quasi-proper
  equilibrium of $\Gamma$. The proof is then concluded by the triangle
  inequality.  \end{proof}

\bibliographystyle{abbrv}
\bibliography{quasiproper}

\end{document}